\documentclass[showpacs,amsmath,amssymb,twocolumn,superscriptaddress,notitlepage,preprintnumbers,prl]{revtex4-1}

\usepackage{graphicx}
\usepackage{dcolumn}
\usepackage{bm}
\usepackage{times}
\usepackage{multirow}
\usepackage{dsfont}
\usepackage{dcolumn}
\usepackage{tabularx}
\usepackage{subfigure}
\usepackage{ae}
\usepackage{lettrine}
\usepackage{color}
\usepackage{amsmath,amsthm,amssymb,amsfonts,boxedminipage}
\usepackage{epstopdf}
\usepackage{hyperref}
\usepackage{dsfont}
\usepackage{amsmath}
\newcommand{\algmargin}{\the\ALG@thistlm}

\makeatother

\newcommand{\norm}[1]{\left\lVert#1\right\rVert}

\newtheorem{theorem}{Theorem}
\newtheorem{lemma}{Lemma}

\newtheorem{definition}{Definition}

\makeatletter
\renewcommand\@biblabel[1]{#1.}
\makeatother

\usepackage[noend,noline,ruled,linesnumbered]{algorithm2e}
\usepackage{hyperref}
\usepackage{bm}
\usepackage{times}
\usepackage{multirow}
\usepackage{dcolumn}
\usepackage{tabularx}
\usepackage{graphicx}
\usepackage{subfigure}
\usepackage{ae}
\usepackage{lettrine}
\usepackage{color}
\usepackage{makecell}

\usepackage{amsmath,amsthm,amssymb,amsfonts,boxedminipage}
\usepackage{epstopdf}

\definecolor{comment}{rgb}{0, 0, 0}

\hypersetup{colorlinks=true, linkcolor=cyan, citecolor=cyan, urlcolor=blue }

\usepackage{braket}

\newcommand{\PKU}{Center on Frontiers of Computing Studies, School of Computer Science, Peking University, Beijing 100871, China}

\begin{document}

\title{Heisenberg limited quantum algorithm for estimating the fidelity susceptibility}

\author{Yukun Zhang}
\affiliation{\PKU}

\author{Xiao Yuan}
\email{xiaoyuan@pku.edu.cn}
\affiliation{\PKU}

\date{\today}

\begin{abstract}
The fidelity susceptibility serves as a universal probe for quantum phase transitions, offering an order-parameter-free metric that captures ground-state sensitivity to Hamiltonian perturbations and exhibits critical scaling. Classical computation of this quantity, however, is limited by exponential Hilbert space growth and correlation divergence near criticality, restricting analyses to small or specialized systems. Here, we present a quantum algorithm that achieves efficient and Heisenberg-limited estimation of fidelity susceptibility through a novel resolvent reformulation, leveraging quantum singular value transformation for pseudoinverse block encoding with amplitude estimation for norm evaluation. This constitutes the first quantum algorithm for fidelity susceptibility with optimal precision scaling. Moreover, for frustration-free Hamiltonians, we show that the resolvent can be approximated with a further quadratic speedup. Our work bridges quantum many-body physics and algorithmic design, enabling scalable exploration of quantum criticality with applications in materials simulation, metrology, and beyond on fault-tolerant quantum platforms.
\end{abstract}

\maketitle

Quantum phase transitions~\cite{sachdev1999quantum} represent a cornerstone of quantum many-body physics, characterizing abrupt changes in the ground-state properties of quantum systems driven by external parameters.  Understanding and detecting these transitions is crucial for characterizing quantum materials, designing quantum devices, and exploring exotic phases of matter. Unlike thermal phase transitions, quantum phase transitions occur at zero temperature and are governed by quantum fluctuations~\cite{sondhi1997continuous}, manifesting in diverse phenomena from superconductivity~\cite{ginzburg2009theory} to topological order~\cite{wen1989vacuum}. A key challenge in studying quantum phase transitions lies in identifying suitable order parameters, which often requires prior knowledge of the system's symmetry-breaking patterns. 

The fidelity susceptibility~\cite{zanardi2006ground,you2007fidelity,gu2010fidelity} has emerged as a versatile, order-parameter-free diagnostic tool for identifying quantum critical points, providing a geometric measure of how rapidly the ground state changes under parameter variations in the system Hamiltonian.
The fidelity susceptibility $\chi_F(\lambda)$ quantifies the rate of change of the ground state $|\Psi_0(\lambda)\rangle$ concerning a control parameter $\lambda$. At quantum critical points, $\chi_F$ exhibits singular behavior, scaling with system size according to universal exponents that characterize the universality class of the transition~\cite{gu2008fidelity,gu2009scaling}. This makes it particularly valuable for finite-size scaling analysis and for distinguishing between different types of phase transitions~\cite{gu2009scaling}. Moreover, the fidelity susceptibility is directly related to quantum Fisher information~\cite{liu2020quantum}, establishing its fundamental role in quantum metrology~\cite{giovannetti2006quantum}.

Despite its theoretical importance, computing the fidelity susceptibility remains a formidable computational challenge. The quantity requires knowledge of the ground state and its response to perturbations, necessitating either full diagonalization of the Hamiltonian or sophisticated many-body techniques. For quantum many-body systems, where the Hilbert space dimension grows exponentially with system size, classical computation quickly becomes intractable~\cite{gu2010fidelity,wang2015fidelity}. This computational barrier has limited the practical application of fidelity susceptibility to small systems or special cases where analytical solutions exist.

The advent of quantum computing offers new possibilities for overcoming these limitations. Quantum algorithms can naturally encode and manipulate quantum states~\cite{nielsen2010quantum}, potentially providing exponential advantages for certain many-body problems. However, designing efficient quantum algorithms for fidelity susceptibility is non-trivial, as it requires information on the full spectrum of the Hamiltonian. While methods based on variational methods~\cite{tan2021variational,di2022quantum} have been proposed, there is no guarantee of performance for the algorithms to work effectively in practical scenarios.

Here, we present quantum algorithms for estimating fidelity susceptibility based on a newly proposed resolvent formulation. It thus relates the estimation of the fidelity susceptibility to the quantum linear system solving problems (QLSP), which is proposed and solved by the seminal work of Ref.~\cite{harrow2009quantum}. Unlike traditional expansion methods that require explicit computation of excited states, our approach reformulates fidelity susceptibility $\chi_F$ as the squared norm of a matrix-vector product involving the Moore-Penrose pseudoinverse of the (shifted) Hamiltonian. This resolvent-based expression enables us to leverage recent advances in quantum singular value transformation (QSVT)~\cite{gilyen2019quantum,martyn2021grand} to construct an efficient block encoding of the relevant operators. Combined with quantum amplitude estimation (QAE)~\cite{brassard2000quantum}, our algorithm achieves Heisenberg-limited scaling~\cite{giovannetti2006quantum,braunstein1994statistical} in the estimation error, requiring $\mathcal{\widetilde{O}}(1/\epsilon)$ queries to estimate $\chi_F$ to additive error $\epsilon$, where $\mathcal{\widetilde{O}}$ notation omits poly-logarithmic scaling. We also uncover further improvements in complexity that could be made for frustration-free systems~\cite{affleck1988valence,sattath2016local}. Furthermore, we discuss applications in the estimation of linear static susceptibility~\cite{kubo1957statistical,fetter2012quantum} and quantum Fisher information~\cite{braunstein1994statistical}.

This work presents the first quantum algorithm for estimating fidelity susceptibility with guaranteed efficiency, resolving a longstanding challenge in its computation and broadening the scope of quantum computing applications.
Beyond this result, it opens pathways to quantum-enhanced studies of critical phenomena, with far-reaching implications for materials design, quantum simulation, and quantum metrology on emerging quantum computing platforms.

\vspace{0.15cm}

\emph{Fidelity Susceptibility}---The fidelity susceptibility quantifies the sensitivity of the ground state to perturbations in a control parameter, serving as a powerful tool for identifying quantum phase transitions without requiring prior knowledge of order parameters~\cite{zanardi2006ground,you2007fidelity,gu2010fidelity}. Consider a Hamiltonian $H(\lambda + \epsilon) = H(\lambda) + \epsilon H_I$, where $H_I$ is the driving term and $\epsilon$ is a small perturbation. The ground-state fidelity is defined as $F(\lambda, \lambda + \epsilon) = |\langle \Psi_0(\lambda) | \Psi_0(\lambda + \epsilon) \rangle|$, where $|\Psi_0(\lambda)\rangle$ denotes the ground state of $H(\lambda)$. The fidelity susceptibility is given by
\begin{equation}
\chi_F(\lambda) = -\left. \frac{\partial^2 \ln F}{\partial \epsilon^2} \right|_{\epsilon=0},
\end{equation}
which measures the second-order response of the fidelity and exhibits peaks or divergences at critical points, following universal scaling laws that characterize the transition's universality class~\cite{gu2008fidelity,gu2009scaling}. Notably, $\chi_F$ is closely related to the quantum Fisher information~\cite{braunstein1994statistical}, establishing its role in quantum metrology where it bounds the precision of parameter estimation. In finite-size systems, $\chi_F$ enables extrapolation of critical points via finite-size scaling, making it invaluable for numerical studies of phase diagrams~\cite{gu2009scaling}.

There are several equivalent formulations of the fidelity susceptibility, each highlighting different aspects and computational approaches. One geometric form expresses it in terms of the ground-state derivative~\cite{campos2007quantum}:
\begin{equation}
\chi_F(\lambda) = \frac{\langle \partial_\lambda \Psi_0 | \partial_\lambda \Psi_0 \rangle}{\langle \Psi_0 | \Psi_0 \rangle} - \frac{|\langle \Psi_0 | \partial_\lambda \Psi_0 \rangle|^2}{\langle \Psi_0 | \Psi_0 \rangle^2},
\end{equation}
which corresponds to the real part of the quantum geometric tensor and connects to Berry curvature in multi-parameter spaces. This formulation is derived by expanding the fidelity $F(\lambda, \lambda + \epsilon) = 1 - \frac{1}{2} \chi_F \epsilon^2 + O(\epsilon^3)$ and expressing the perturbed state as $|\Psi_0(\lambda + \epsilon)\rangle \approx |\Psi_0(\lambda)\rangle + \epsilon |\partial_\lambda \Psi_0\rangle$, without assuming normalization of the derivative term. A perturbative expansion yields~\cite{gu2010fidelity}
\begin{equation}\label{eq:perturb}
\chi_F(\lambda) = \sum_{j \neq 0} \frac{|\langle \Psi_j(\lambda) | H_I | \Psi_0(\lambda) \rangle|^2}{[E_j(\lambda) - E_0(\lambda)]^2},
\end{equation}
which involves the full eigenspectrum and revealing contributions from excited states, particularly dominant near gap-closing transitions. This is formulated using time-independent perturbation theory, where the first-order correction to the ground state involves summing over matrix elements divided by energy differences, leading to the quadratic denominator after fidelity expansion. Finally, a linear-response formulation~\cite{you2007fidelity} recasts it as
\begin{equation}
\chi_F(\lambda) = \int_0^\infty d\tau \frac{\langle \Psi_0 | H_I(\tau) H_I | \Psi_0 \rangle - \langle \Psi_0 | H_I | \Psi_0 \rangle^2}{\tau},
\end{equation}
where $H_I(\tau) = e^{H \tau} H_I e^{-H \tau}$, amenable to time-evolution simulations. This representation is derived by interpreting the perturbative sum as the zero-frequency component of a spectral function and applying a Fourier transform. These representations underscore the multifaceted nature of $\chi_F$, linking static ground-state properties to dynamic responses.

Classically, fidelity susceptibility $\chi_F$ has been computed using exact diagonalization for small systems~\cite{gu2010fidelity}, density-matrix renormalization group methods for one-dimensional and quasi-1D systems~\cite{tzeng2008scaling}, and quantum Monte Carlo techniques~\cite{wang2015fidelity}, which can efficiently sample the linear-response form in sign-problem-free models. However, these classical methods apply only to special cases, and because all definitions of $\chi_F$ involve nontrivial features of the Hamiltonian, computing $\chi_F$ in the general setting is classically inefficient.

\vspace{0.15cm}

\emph{Main Results---}We now show that quantum computers can bypass these classical limitations and efficiently estimate the fidelity susceptibility. 

Our analysis is based on the following assumptions:
(1) the Hamiltonian $H(\lambda)$ has a nondegenerate ground state $|\Psi_0(\lambda)\rangle$;
(2) a lower bound $\Delta$ on the spectral gap is known, such that $\Delta \leq E_1(\lambda) - E_0(\lambda)$;
(3) we have access to a unitary (and its inverse) that prepares the ground state, $U|0^n\rangle = |\Psi_0\rangle$, along with the corresponding ground-state energy $E_0$.
The first two assumptions are standard in studies of quantum phase transitions and ground-state preparation~\cite{lin2020near}. Regarding the third assumption, since $H(\lambda)$ is typically a simple Hamiltonian, it is reasonable to assume that the ground state can be prepared via a simple, even classical unitary process. Otherwise, one may rely on well-established quantum algorithms such as adiabatic state preparation~\cite{albash2018adiabatic}, or quantum phase estimation~\cite{kitaev1995quantum}. The corresponding ground-state energy $E_0$ can then be obtained either through quantum phase estimation~\cite{kitaev1995quantum} or by directly measuring the expectation value of the Hamiltonian. We note that extensions to degenerate cases~\cite{su2013quantum}, relevant to symmetry-broken or topological phases~\cite{wen1989vacuum}, would require explicit characterization of the ground-state subspace, which we leave for future work. For notational simplicity, we omit the explicit $\lambda$ dependence in what follows.

The cornerstone of our approach is a novel resolvent-based reformulation of the fidelity susceptibility, expressed as
\begin{equation}\label{eq:key_formula}
\begin{aligned}
    \chi_F &= \langle \Psi_0 | H_I P^\perp (H - E_0)^{-2} P^\perp H_I | \Psi_0 \rangle \\
    &= \| P^\perp (H - E_0)^{-1} P^\perp H_I | \Psi_0 \rangle \|_2^2 \\
    &= \| (H - E_0)^+ H_I | \Psi_0 \rangle \|_2^2 =: \| G | \Psi_0 \rangle \|_2^2,
\end{aligned}
\end{equation}
where $P^\perp = I - P$, $P = | \Psi_0 \rangle \langle \Psi_0 |$, $\| \cdot \|_2$ denotes the $\mathrm{L}_2$ norm, and $A^+$ is the Moore-Penrose pseudoinverse of $A$.
The formula is equivalent to the perturbative expansion given by Eq.~\eqref{eq:perturb}.  Our expression is verified by spectral decomposition of the resolvent $(H - E_0)^{-1}$ in the Hamiltonian's eigenbasis, where the pseudoinverse naturally excludes the zero-eigenvalue ground state while inverting the excited-state contributions. The projector $P^\perp$ ensures orthogonality to the ground state, validating the inversion; equivalently, the pseudoinverse sets the zero eigenvalue contribution to zero while reciprocating others.

This resolvent formulation offers several critical advantages over traditional expressions. Unlike the perturbative sum, which requires explicit computation or summation over potentially exponentially many excited states, the resolvent approach encapsulates all contributions implicitly through matrix inversion—a task amenable to quantum linear algebra techniques~\cite{harrow2009quantum,gilyen2019quantum}. This shift avoids the need for full diagonalization, which is intractable classically for large systems, and enables efficient approximation via polynomial series. Furthermore, it establishes a direct connection to QLSPs~\cite{harrow2009quantum}, but rather focuses on realizing the matrix inversion rather than the preparation of states. The implications are twofold: computationally, it allows for polynomial-time quantum estimation in parameters like the spectral gap $\Delta$; conceptually, it bridges fidelity susceptibility to resolvent methods in operator theory, potentially inspiring new analytical bounds on critical scaling. To our knowledge, this is the first application of resolvent techniques to fidelity susceptibility computation, enabling efficient quantum implementation by reducing the problem to matrix inversion and norm estimation, bypassing explicit excited-state summations that hinder classical and prior quantum approaches. Crucially, the matrix inversion problem is shown to be a BQP-complete problem~\cite{harrow2009quantum}. This indicates that the estimation of fidelity susceptibility as a quantum-physics-related problem could exhibit a quantum advantage over classical computers.

Under the above stated assumptions, we obtain the following main result, marking the first Heisenberg-limited quantum algorithm for this task and demonstrating a clear quantum advantage in precision and scalability.

\begin{theorem}[Informal]\label{theorem:main}
Assume access to block encodings $U_H$ and $U_I$ of $H - E_0$ and $H_I$ with normalization factors $\alpha_H$ and $\alpha_I$, respectively, and access to the ground state $|\Psi_0\rangle$. Then, there exists a quantum algorithm that estimates $\chi_F$ to additive error $\epsilon$ with high probability using $\mathcal{O}\left(\frac{\alpha_H \alpha_I^2}{\Delta^3 \epsilon} \log \left(\frac{\alpha_I}{\Delta \epsilon}\right)\right)$ queries to $U_H$ and its inverse, and $\mathcal{O}\left(\frac{\alpha_I^2}{\Delta^2 \epsilon}\right)$ queries to $U_I$ and its inverse.
\end{theorem}

The Heisenberg scaling in $\epsilon$ arises from the quadratic precision boost of quantum amplitude estimation, while the $\Delta^{-3}$ dependence reflects the condition number of the inversion amplified by the norm estimation. 
Here we briefly summarize the quantum algorithm as follows and refer to Supplementary Materials for details.

Step 1. 
Construct a block encoding of the pseudoinverse $(H - E_0)^+$ by approximating the reciprocal function $1/x$ via QSVT, which applies a polynomial approximation to the singular (eigen) values of the normalized Hamiltonian. This step leverages Chebyshev polynomials for near-optimal approximation, with degree scaling as $\mathcal{O}(\alpha_H / \Delta \log(1/\epsilon'))$ for internal error $\epsilon'$.
    
Step 2. Form the block encoding of $G = (H - E_0)^+ H_I$ by taking the product of the block encoded matrix $(H - E_0)^+$ and $H_I$.
    
Step 3. Estimate $\chi_F = \|G |\Psi_0\rangle \|_2^2$ using QAE methods applied to the action of the block-encoded $G$ on $|\Psi_0\rangle$, which provides a quadratic speedup in precision over naive sampling.

We remark that as long as the lower bound on the spectral gap remains non-vanishingly (polynomially) small, our algorithm remains efficient. Our framework not only achieves optimal Heisenberg scaling in $\epsilon$ but also optimizes dependence on $\Delta$ and operator norms, representing a significant advance over classical methods limited by exponential scaling or perturbative approximations.

\vspace{0.15cm}

For certain special classes of Hamiltonians, we can achieve a quadratic improvement in the dependence on the spectral gap $\Delta$ in the matrix inversion step. In particular, consider frustration-free (FF) Hamiltonians~\cite{affleck1988valence,sattath2016local}, defined as $H_F = \sum_{j=0}^{r-1} \Pi_j$ where each $\Pi_j$ is a projector and the projectors share a common ground space (i.e., $\cap_j \ker(\Pi_j) \neq \emptyset$), with ground energy $E_0 = 0$. Such Hamiltonians arise in various quantum many-body models, including quantum satisfiability problems~\cite{bravyi2006efficient}, certain topological phases~\cite{kitaev1995quantum}, and exhibit unique properties like Parent Hamiltonians for PEPS~\cite{perez2007peps} or exact solvability in some cases~\cite{affleck1988valence}. The key property enabling speedup is that FF Hamiltonians admit a structured block encoding where the ground state maps to eigenvalue one~\cite{somma2013spectral,king2025quantum}, allowing polynomial approximations of the inverse to exploit rapid oscillations near the spectral edge via Chebyshev polynomial of the first kind~\cite{schaeffer1941inequalities,trefethen2019approximation}, reducing the degree from $\mathcal{O}(\alpha_H / \Delta)$ to $\mathcal{O}(\sqrt{r / \Delta})$ where $r$ bounds the number of terms and also the operator norm of $H_F$. Similar results have been previously observed in several works on FF systems~\cite{somma2013spectral,gosset2016correlation,thibodeau2023nearly}, and most importantly, have also been proposed for special types of positive definite QLSPs~\cite{orsucci2021solving}. Consequently, for FF Hamiltonians, our algorithm yields:

\begin{theorem}[Informal]\label{theorem:ff}
For FF Hamiltonians $H_F$ with $r$ terms, assuming access to block encodings as above, there exists a quantum algorithm that estimates $\chi_F$ to additive error $\epsilon$ with high probability using $\mathcal{O}\left(\frac{\sqrt{r} \alpha_I^2}{\Delta^{2.5} \epsilon} \log \left(\frac{r \alpha_I}{\Delta \epsilon}\right)\right)$ queries to the block encoding of $H_F$ and its inverse, and $\mathcal{O}\left(\frac{\alpha_I^2}{\Delta^2 \epsilon}\right)$ queries to $U_I$ and its inverse.
\end{theorem}

This improvement highlights how Hamiltonian structure can be exploited for algorithmic gains, similar to speedups in ground-state preparation~\cite{thibodeau2023nearly} or spectral gap amplification~\cite{somma2013spectral} for FF systems, and suggests avenues for tailoring algorithms to specific physical models.

\vspace{0.15cm}

\emph{Applications---}Our quantum algorithm for estimating the fidelity susceptibility naturally extends to several related quantities in quantum many-body physics and metrology, each of which admits a resolvent formulation amenable to the same block-encoding and amplitude-estimation techniques. Below, we discuss two prominent examples: linear static susceptibilities and quantum Fisher information. We provide a brief background, the original definition (typically a sum over excited states), and a derivation of the resolvent form in the main text.

Linear static susceptibilities~\cite{kubo1957statistical,fetter2012quantum} quantify the response of a system to external static perturbations, such as magnetic fields (magnetic susceptibility) or electric fields (charge or dielectric susceptibility). They play a crucial role in characterizing equilibrium properties, phase transitions, and material responses in condensed matter physics. For instance, the magnetic susceptibility diverges at critical points in ferromagnetic transitions, serving as an order parameter diagnostic~\cite{stanley1971phase,fisher1998renormalization}.
The original definition arises from linear response theory, where the expectation value of an observable $\hat{O}$ (e.g., total spin $S_z$ for magnetic susceptibility) responds to a perturbation $-h \hat{O}$ added to the Hamiltonian. Defining the static susceptibility by
\(\chi:=\partial\langle\hat O\rangle/\partial h\big|_{h=0}\), one obtains for a nondegenerate ground state \(|\Psi_0\rangle\) the standard spectral (Lehmann) sum at \(T=0\):
\begin{equation}\label{eq:chi_zeroT}
\chi \;=\; 2\sum_{j\neq0}\frac{|\langle\Psi_j|\hat O|\Psi_0\rangle|^2}{E_j-E_0}\,,
\end{equation}
where \(H|\Psi_j\rangle=E_j|\Psi_j\rangle\).
As such, we find the linear static susceptibility shares a similar form to the fidelity susceptibility, with the only change from $H_I$ to $\hat{O}$ and an overall factor $2$. Our algorithm thus applies straightforwardly. 

The quantum Fisher information (QFI)~\cite{braunstein1994statistical} quantifies the maximum precision in estimating a parameter $\lambda$ encoded in a quantum state, via the quantum Cram\'er-Rao bound~\cite{braunstein1994statistical} $\Delta \lambda \geq 1 / \sqrt{\nu F(\lambda)}$, where $\nu$ is the number of measurements. For pure states, it relates directly to state sensitivity and metrological advantage in sensors~\cite{paris2009quantum} and the corresponding sum form is 
\begin{equation}
F(\lambda) = 4 \sum_{j \neq 0} \frac{|\langle \Psi_j | \partial_\lambda H | \Psi_0 \rangle|^2}{(E_j - E_0)^2} = 4 \chi_F,
\end{equation}
where $\chi_F$ is the fidelity susceptibility if the driving is $\partial_\lambda H$. As such, our method also provides a way to estimate the QFI.

\vspace{0.15cm}


\emph{Conclusion}---We have presented a quantum algorithm for estimating the fidelity susceptibility of quantum many-body systems, achieving Heisenberg-limited precision with query complexity $\mathcal{\widetilde{O}}(\alpha_H\alpha_I^2/\Delta^3\epsilon)$. Our approach, based on a novel resolvent formulation of the fidelity susceptibility, circumvents the need for explicit diagonalization or excited state computation that limits classical methods. By leveraging quantum singular value transformation for matrix inversion and quantum amplitude estimation for norm computation, we provide an end-to-end quantum protocol suitable for implementation on fault-tolerant quantum computers.
A particularly striking result is the quadratic speedup achieved for frustration-free Hamiltonians for matrix inversion, where the query complexity improves to $\mathcal{\widetilde{O}}(\sqrt{r}\alpha_I^2/\Delta^{2.5}\epsilon)$. This enhancement stems from the special spectral properties of frustration-free systems, which allow for more efficient polynomial approximations of the inverse function near the spectral edge. This finding adds to the growing body of evidence that frustration-free systems exhibit fundamentally different computational properties, joining other examples such as quadratic improvements in spectral gap amplification and correlation length bounds.

Our work opens several avenues for future research. The extension to degenerate ground states, relevant for topological phases and symmetry-protected ordered phases, requires careful treatment of the projector onto the ground state manifold. The algorithm could be adapted to compute other response functions and susceptibilities by modifying the driving operator $H_I$. 
From a practical perspective, our algorithm provides a path toward quantum advantage in studying phase transitions in strongly correlated systems. As quantum hardware continues to improve, the ability to compute fidelity susceptibility for systems beyond classical reach will enable the discovery and characterization of novel quantum phases of matter. The connection to quantum Fisher information also suggests applications in quantum sensing and parameter estimation, where the fidelity susceptibility determines fundamental precision limits.

\vspace{0.15cm}

\emph{Acknowledgments}---This work is supported by the National Natural Science Foundation of China Grant (No.~12361161602), NSAF (Grant No.~U2330201), the Innovation Program for Quantum Science and Technology (Grant No.~2023ZD0300200), and the High-performance Computing Platform of Peking University.

\bibliography{ref}
\clearpage

\widetext
\section*{Supplementary Information}
\appendix

\section{Quantum Singular Value Transformation}
In this section, we introduce the QSVT algorithm. The algorithm provides a unified framework for designing quantum algorithms. The algorithm performs a singular mapping for a given matrix: $A\mapsto f^{(\mathrm{SV})}(A)$ with $f(x)$ some analytic function. For normal matrices such that $[A,A^\dagger]=0$, the transformation reduces to the eigenvalues. It is thus seen why such a framework can unify widespread quantum algorithms~\cite{martyn2021grand}, since most quantum algorithms boil down to a task for performing some matrix function.

The QSVT algorithm essentially works in the language of block encoding (BE), which is defined as:
\begin{definition}[Block Encoding]
\label{def:block_encoding}
Let $A \in \mathbb{C}^{2^n \times 2^n}$ be a matrix. A unitary operator $U \in \mathbb{C}^{2^{m+n} \times 2^{m+n}}$ is called an $(\alpha_A, m, \epsilon)$-block encoding of $A$ if
\begin{equation}\label{eq:be}
\left\| A - \alpha_A \left( \bra{0^m} \otimes I_n \right) U_A \left( \ket{0^m} \otimes I_n \right) \right\| \leq \epsilon,
\end{equation}
where $\alpha_A \geq \|A\|$ is the normalization factor, $m$ is the number of ancilla qubits, and $\epsilon \geq 0$ is the encoding error.
\end{definition}

The block encoding allows flexible operations, such as taking the product of two block encoded matrices, as shown below.
\begin{lemma}[Product of two block-encoded matrices {\cite[Adapted from Lemma 53]{gilyen2019quantum}}]\label{lemma:block_prod}
Let $U_A$ be a $(\alpha_A, a, \epsilon_1)$-block encoding of $A$. Let $U_B$ be a $(\alpha_B, b, \epsilon_2)$-block encoding of $B$. Then, $(I_b\otimes U_A)(I_a \otimes U_B)$ is a $(\alpha_A\alpha_B, a+b, \alpha_A\epsilon_2+\alpha_B\epsilon_1)$-block encoding of $AB$.
\end{lemma}


Finally, we provide the main results for performing matrix mapping from Ref.~\cite{gilyen2019quantum}.
\begin{lemma}[matrix function of real polynomials {\cite[Adapted from Corollary 18]{gilyen2019quantum}}]\label{lemma:qsvt}
Let $U$ be a $(\alpha_A, a, 0)$-block encoding of matrix $A$. Let $p_{\Re} \in \mathbb{R}[x]$ be a degree-$m$ polynomial function such that
\begin{itemize}
    \item $p_{\Re}$ has parity $(m~\text{mod}~2)$ and
    \item $|p_{\Re}|\leq 1,~\forall x\in[-1,1]$.
\end{itemize}
Then, there exists an algorithm construct the $(1,a+1,0)$-block encoding of $p_{\Re}^{\mathrm{(SV)}}(A/\alpha_A)$ with $p_{\Re}^{\mathrm{(SV)}}$ denote the function acting on singular values using $m$ queries to $U$, $U^\dagger$ and $\mathcal{O}((a+1)m)$ other primitive quantum gates.
\end{lemma}

\section{Main results}

\subsection{Block encoding of the $G$ operator}
The most involved part for block encoding is the inverse of $(H-E_0)$. The inverse of a matrix is well-studied throughout the history of quantum computation, as this is a basic problem for the QLSP~\cite{harrow2009quantum}. That is, given $A\ket{x}=\ket{b}$, we want to (approximately) solve for $\ket{x}$ within error $\epsilon$ by inverting $A$. The crucial quantity is the condition number of $A$, defined as $\kappa(A)=\norm{A}\norm{A^{-1}}$.

The seminal HHL algorithm~\cite{harrow2009quantum} inaugurated this study and proved BQP-hardness for the problem. A long line of work applying different strategies, such as adiabatic evolution~\cite{subacsi2019quantum,costa2022optimal}, the quantum Zeno effect~\cite{an2022quantum}, and QSVT~\cite{gilyen2019quantum}. It wasn't until recently that the optimal quantum algorithm was proposed, with the query complexity of BE of $A$ as $\Theta(\kappa\log(1/\epsilon))$, utilizing a refined discrete adiabatic formula~\cite{costa2022optimal}.

The main difference between our scenario and the QLSP is that we do not need to perform the matrix inversion on a state vector, but only block encode the inversion. As such, we consider the QSVT algorithm to directly approximate the $\frac{1}{x}$ function for our input matrix.

We assume that we have access to the $(\alpha_H, m, 0)$-block encoding of $(H-E_0)$ such that 
\begin{equation}
    \alpha_H\geq \norm{H-E_0}.
\end{equation}
Here, $m$ is the number of ancilla qubits. In principle, we could obtain such a BE through the linear combination of unitaries (LCU) method~\cite{childs2012hamiltonian}, since the Hamiltonian can be expanded in the Pauli basis. Besides, because we have access to the ground state $\ket{\Psi_0}$ of $H$, using methods like quantum phase estimation, we can get an estimation of $E_0$ to an arbitrary desirable accuracy. For simplicity, we here assume that we know $E_0$ exactly.

\begin{lemma}[Polynomial approximation of $\frac{1}{x}$ {\cite[Corollary 69]{gilyen2019quantum}}]\label{lemma:inverse_func}
Let $\varepsilon, \delta \in\left(0, \frac{1}{2}\right]$. There exists an odd polynomial function $p_d(x)\in\mathbb{R}(x)$ of degree
\begin{equation}
    d=\mathcal{O}\left(\frac{1}{\delta} \log \left(\frac{1}{\varepsilon}\right)\right)
\end{equation}
such that
\begin{equation}
    \sup_{x\in T}|p_d(x)-f(x)|\leq \varepsilon,
\end{equation}
where
\begin{equation}\label{eq:inverse_interval}
    x\in[-1,1] \backslash [-\delta, \delta],
\end{equation}
and $f(x)=\frac{3}{4} \frac{\delta}{x}$. Besides, $p_d(x)$ is bounded by $1$ in absolute value.
\end{lemma}

Using the above lemma, we have the following results regarding the block encoding of a matrix inverse.

\begin{lemma}[Block encoding of $A^{+}$]\label{lemma:block_inverse}
Let $\epsilon, \delta \in\left(0, \frac{1}{2}\right]$. Given a $(\alpha_A,m,0)$-block encoding $U_A$ of a positive semi-definite $n$-qubit operator $A$ with $\alpha_A\geq \norm{A}$, denote its smallest nonzero eigenvalue as $\eta$. There exists a quantum circuit that implements the $(\alpha_{A'},m+1,0)$-block encoding unitary $U_p$ of $p_d(\alpha_A^{-1} A)$ such that
\begin{equation}
    \norm{\alpha_{A'}\bra{0^{m+1}} U_p\ket{0^{m+1}}-A^+}\leq \epsilon,
\end{equation}
where $A^+$ is the Moore-Penrose inverse of $A$, $p_d(x)\in\mathbb{R}(x)$ is an odd polynomial function of degree
\begin{equation}
    d=\mathcal{O}\left(\frac{\alpha_A}{\eta} \log \left(\frac{1}{\eta\epsilon}\right)\right),
\end{equation}
and $\alpha_{A'}=\frac{4}{3\eta}$ using $d$ queries to $U_A$ and $U_A^\dagger$ along with $\mathcal{O}(d(m+n+1))$ elementary quantum gates.
\end{lemma}
\begin{proof}
We note that the matrix inside $U_A$ is $\alpha_A^{-1} A$ such that
$$U_A=
\left(\begin{array}{cc}
    \alpha_A^{-1} A & * \\
    * & *
\end{array}\right).
$$
The smallest nonzero eigenvalue of the block encoded matrix is $\alpha_A^{-1}\eta$, which determines the spectral radius of the inverse matrix. 

Consider the function given by Lemma \ref{lemma:inverse_func} with the input being the matrix $\alpha_A^{-1} A$, we have
\begin{equation}
    \norm{p_d(\alpha_A^{-1} A)-f(\alpha_A^{-1} A)}=\norm{p_d(\alpha_A^{-1} A)-\frac{3}{4}\eta A^+}=\sup_{x\in T}|p_d(x)-f(x)|\leq \varepsilon,
\end{equation}
where $f(x)=\frac{3}{4}\frac{\delta}{x}$. For the polynomial approximation, the $\delta$ in Eq.~\eqref{eq:inverse_interval} is chosen to be $\delta=\alpha_A^{-1}\eta$. Here, the polynomial function approximates the Moore-Penrose inverse because it is an odd function, so that $p_d(0)=0$.

Now, when transforming to the block encoding, the approximation error becomes
\begin{equation}
    \norm{\alpha_{A'}\bra{0^{m+1}} U_p\ket{0^{m+1}}-A^+}=\norm{\alpha_{A'} p_d(\alpha_A^{-1} A)-A^+}\leq \alpha_{A'}\varepsilon.
\end{equation}

Finally, to let $\alpha_{A'}\varepsilon\leq \epsilon$, we have 
\begin{equation}
    \varepsilon \leq \frac{3}{4}\eta \epsilon.
\end{equation}
This explains the claimed degree and complexity.
\end{proof}

We note that the block encoding unitary of the polynomial function is given by
\begin{equation}
U_p=
\left(\begin{array}{cc}
    \alpha_{A'}^{-1} p_d(\alpha_A^{-1} A) & * \\
    * & *
\end{array}\right).
\end{equation}

Our next step is to show the block encoding of $G$.
\begin{theorem}[Block encoding of $G$]\label{theorem:be_g}
Let $\epsilon\in\left(0,\frac{1}{2}\right]$. Let $U_H$ be an $(\alpha_H, m_1,0)$-block encoding of $n$-qubit operator $(H-E_0)$. Let $U_I$ be an $(\alpha_I, m_2,0)$-block encoding of $n$-qubit operator $H_I$. Suppose that $H$ has a nondegenerate ground state, and we have knowledge of the lower bound $\Delta$ of the spectral gap such that $\Delta\leq E_1-E_0$.

Then, there exists a quantum circuit that implements an $(\alpha_Q,m_1+m_2+1,0)$-block encoding $U_Q$ of $Q$ such that
\begin{equation}
    \norm{\alpha_Q\bra{0^{m_1+m_2+1}} U_Q \ket{0^{m_1+m_2+1}} -G}\leq \epsilon,
\end{equation}
where $\alpha_Q=\frac{4\alpha_I}{3\Delta}$ and $G=(H-E_0)^+H_I$, using
\begin{equation}
    \mathcal{O}\left(\frac{\alpha_H}{\Delta} \log \left(\frac{\alpha_I}{\Delta\epsilon}\right)\right)
\end{equation}
queries to $U_H$ and its inverse along with one query to $U_I$ and $\mathcal{O}\left((n+m_1+m_2+1)\frac{\alpha_H}{\Delta} \log \left(\frac{\alpha_I}{\Delta\epsilon}\right)\right)$ elementary gates.
\end{theorem}
\begin{proof}
The first step of the algorithm is to construct the block encoding approximating $(H-E_0)^+$. This is achieved by Lemma \ref{lemma:block_inverse}. That is, for $\epsilon'\in\left(0,\frac{1}{2}\right]$, we can construct a unitary block encodes $\alpha_{H'}^{-1} p_d(\alpha_H^{-1} (H-E_0))$ with 
\begin{equation}
    d=\mathcal{O}\left(\frac{\alpha_H}{\Delta} \log \left(\frac{1}{\Delta\epsilon'}\right)\right)
\end{equation}
obeying
\begin{equation}
    \norm{p_d(\alpha_H^{-1} (H-E_0))-(H-E_0)^+} \leq \epsilon'
\end{equation}
using $d$ queries to $U_H$ and its inverse. Here, $\alpha_{H'}=\frac{4}{3\Delta}$.

The next step is to take the product between the above block encoded matrix and $H_I$. By Lemma \ref{lemma:block_prod}, we can construct an $(\alpha_{H'}\alpha_I, m_1+m_2+1, \alpha_I\epsilon')$-block encoding of $G$ with one query to each of the block encoding. Then, let $\alpha_I\epsilon'\leq \epsilon$, and we obtain the claimed results.
\end{proof}
Here, we note that $U_Q$ is a block encoding of $Q$, but $Q$ also approximates $G$. Therefore, we will refer to $U_Q$ as the block encoding of $Q$ and $G$ interchangeably, but with the approximation error varying.

\subsection{Estimation of the fidelity susceptibility}
We have shown methods to construct the approximate block encoding of $G$. Our main algorithm is to apply the quantum amplitude estimation method for estimating the fidelity susceptibility. To this end, we introduce the QAE algorithm.
\begin{lemma}[Quantum amplitude estimation, Ref.~\cite{brassard2000quantum}]\label{lemma:ae}
Suppose $U$ is an unitary operation acting on register $a$ and $b$ such that
$$
U|0\rangle_{ab}=\sqrt{p}|0\rangle_a\left|\phi\right\rangle_b+\sqrt{1-p}|1\rangle_a\left|\phi'\right\rangle_b,
$$
where $\left|\phi\right\rangle$ and $\left|\phi'\right\rangle$ are pure quantum states and $p \in[0,1]$. There exists a quantum algorithm that outputs an estimation $\tilde{p} \in[0,1]$ such that
$$
|\tilde{p}-p| \leq \frac{2 \pi \sqrt{p(1-p)}}{K}+\frac{\pi^2}{K^2}
$$
with probability at least $\frac{8}{\pi^2}$, using $\mathcal{O}(K)$ queries to $U$ and $U^{\dagger}$.
\end{lemma}
In summary, given access to the unitary $U$ and its inverse, we can estimate the square of the amplitude flagged by a certain output of the ancilla register to an additive error $\varepsilon$ using $\mathcal{O}(1/\varepsilon)$ queries to $U$ and its inverse.

We have all the elements of our main algorithm in place, and we present the results as follows.
\begin{theorem}[Estimation of the fidelity susceptibility, formal version of Theorem \ref{theorem:main}]\label{theorem:est_fid}
Let $\epsilon\in\left(0,\frac{1}{2}\right]$. Let $\chi_F=\norm{G\ket{\Psi_0}}_2^2=\norm{\left(H-E_0\right)^{+}H_I\ket{\Psi_0}}_2^2$.
There exists an algorithm outputs with probability at least $\frac{8}{\pi^2}$ an estimation ${\chi}'_F$ such that 
$$|\chi'_F-\chi_F|\leq \epsilon$$
using 
\begin{equation}
    \mathcal{O}\left(\frac{\alpha_H\alpha_I^2}{\Delta^3 \epsilon} \log \left(\frac{\alpha_I}{\Delta\epsilon}\right)\right)
\end{equation}
queries to $U_H$ and its inverse, and 
\begin{equation}
    \mathcal{O}\left(\frac{\alpha_I^2}{\Delta^2\epsilon} \right)
\end{equation}
queries to $U_I$ and its inverse.
\end{theorem}

\begin{proof}
The input state to our algorithm is the ground state $\ket{\Psi_0}$ of $H$ along with the ancilla register set to the all-zero state. Our algorithm proceeds by constructing the approximate block encoding of $G$. Invoking Theorem \ref{theorem:be_g}, for error $\epsilon_1$, we can construct an $(\alpha_{H'}\alpha_I, m_1+m_2+1, \epsilon_1)$-block encoding of $G$ with
\begin{equation}
\mathcal{O}\left(\frac{\alpha_H}{\Delta} \log \left(\frac{\alpha_I}{\Delta\epsilon_1}\right)\right)
\end{equation}
queries to $U_H$, its inverse and one query to $U_I$.

The next step is to apply QAE to estimate the fidelity susceptibility. That is, we perform the block encoding of $Q$ on $\ket{0^{m_1+m_2+1}}\ket{\Psi_0}$ and then estimate the probability of successfully implementing the block encoding flagged by the output of the ancilla register to be in state $\ket{0^{m_1+m_2+1}}$. As the block encoded operator $Q$ approximates $G$ is normalized, we estimate the normalized quantity to an accuracy $\epsilon_2$ and then multiply the normalization factor. Denote the direct output of the QAE algorithm as $E$. Then, we have
\begin{equation}
    \left| E-\norm{\alpha_Q^{-1}Q\ket{\Psi_0}}_2^2 \right|\leq \epsilon_2,
\end{equation}
where $\norm{\alpha_Q^{-1}Q\ket{\Psi_0}}_2^2$ is the success probability for implementing the block encoding of $Q$. As stated, the final estimation is given by $\alpha_QE$, we require
\begin{equation}
\begin{aligned}
    |\alpha_Q^2E-\chi_F|&\leq \epsilon,
\end{aligned}
\end{equation}
where $\chi'_F:=\alpha_Q^2E$ is the final estimation of the fidelity susceptibility.
By inserting the success probability and applying the triangle inequality, we have a sufficient condition to be
\begin{equation}
\begin{aligned}
    \left|\alpha_Q^2 E-\alpha_Q\norm{\alpha_Q^{-1}Q\ket{\Psi_0}}_2^2+\alpha_Q^2\norm{\alpha_Q^{-1}Q\ket{\Psi_0}}_2^2-\chi_F\right| &\leq \epsilon\\
    \left|\alpha_Q^2 E-\alpha_Q^2\norm{\alpha_Q^{-1}Q\ket{\Psi_0}}_2^2\right| +\left|\alpha_Q^2\norm{\alpha_Q^{-1}Q\ket{\Psi_0}}_2^2-\chi_F\right| &\leq \epsilon
\end{aligned}.
\end{equation}
It suffices to prescribe the two sources of error to each contribute $\frac{\epsilon}{2}$. This dictates $\alpha_Q^2\epsilon_2\leq \frac{\epsilon} {2}$, which leads to $\epsilon_2\leq \frac{\epsilon}{2\alpha_Q^2}$. Besides, because $\norm{Q-G}\leq \epsilon_1$, we obtain 
\begin{equation}
\begin{aligned}
    \left| \norm{Q\ket{\Psi_0}}_2^2-\norm{G\ket{\Psi_0}}_2^2\right|&=(\norm{Q\ket{\Psi_0}}_2+\norm{G\ket{\Psi_0}}_2)\left| \norm{Q\ket{\Psi_0}}_2-\norm{G\ket{\Psi_0}}_2\right|\\
    &\leq 2\alpha_Q\norm{Q-G} \leq 2\alpha_Q\epsilon_1,
\end{aligned}
\end{equation}
where we have used $\norm{A\ket{x}}_2\leq \norm{A}$ and $\alpha_Q$ is an upper bound on both $\norm{Q}$ and $\norm{G}$. Let $2\alpha_Q\epsilon_1\leq\frac{\epsilon}{2}$, giving $\epsilon_1\leq \frac{\epsilon}{4\alpha_Q}$. 

Plugging in the $\epsilon_1$ and $\epsilon_2$ to Theorem \ref{theorem:be_g} and Lemma \ref{lemma:ae} yields the claimed results. 
Besides, we remark that the QAE algorithm repeatedly makes queries to $U_Q$, which explains the final results. This finishes the proof.
\end{proof}

\subsection{Quadratic speedup for special types of Hamiltonians}
The matrix inversion is a crucial step in the algorithm presented in the last section. Here, because the matrix we considered $(H-E_0)$ is positive semi-definite, it is found that in certain cases we may improve the dependence on the condition number from linear to square root, indicating a quadratic speedup. It should be remarked that even for general positive definite matrices, a $\Omega(\kappa)$ lower bound is proven in Ref.~\cite{orsucci2021solving}. For Hamiltonians, the frustration-free (FF) Hamiltonian is one of the subclasses of shows such a quantum speedup. 

A FF Hamiltonian is defined as $H_F=\sum_{j=0}^{r-1} \Pi_j$ for $\Pi_j$ are projectors, and $\cap_i \Pi\neq \emptyset$, so that the ground state energy $E_0=0$. In this work, we consider the ground state to be nondegenerate. Intriguingly, several different features are quadratically improvable for the FF Hamiltonian compared to general systems, such as the spectral gap amplification (SGA)~\cite{somma2013spectral} and correlation length quadratically improvement~\cite{gosset2016correlation}.

We give an intuition of why such a quadratic speedup may only present to the FF systems. From a constructive perspective, it is possible to construct a $(1,m,0)$-block encoding of $H'=(I-\alpha_F^{-1} H_F)$ when $H$ is frustration-free with $\alpha_F^{-1}\geq \norm{H_F}$ as a normalization factor. It is crucial to have the normalization factor to be one, so that the spectral radius of the block encoded matrix is $1$, corresponding to the ground state energy. Besides, the first excited state component is mapped to eigenvalue $1-\alpha_F^{-1}E_1$. Now, consider that we want to perform a matrix (analytic) function $f(x)$ of $H'$. If $f(x)$ has its largest derivative at $x=1$, then the cost of approximating such a function is quadratically improved over those that do not. This is seen by approximation theory~\cite{sachdeva2014faster}, such that there exist polynomial functions that oscillate more rapidly at $x=\pm 1$, resulting in quadratic improvements in the truncation degree. The behavior is also justified by Bernstein's inequality~\cite{schaeffer1941inequalities}, which lower bounds the derivative of polynomial functions: i) $\Omega(d)$ for $x$ is $\mathcal{O}(1)$ away from $\pm 1$; and ii) $\Omega(d^2)$ for $x=\pm1$ for degree $d$ polynomial functions. Such a behavior is saturated by the Chebyshev polynomial of the first kind~\cite{trefethen2019approximation}.

The two ingredients to achieve the speedup are i) the block encoded matrix has spectral radius one; and ii) the matrix function of the block encoded matrix has the largest derivative at $x=\pm1$. The FF Hamiltonians meet the former requirement, and we will constructively show how to realize the claimed block encoding. The second requirement is seen by previous studies, such as the preparation ground state of FF systems~\cite{thibodeau2023nearly}, where we need to realize a filter function to filter out other excited-state contributions. The approximation using polynomials of degree $\mathcal{O}(\Delta^{-0.5})$, a quadratic enhancement over the general case. The speed up is gradually lost when the block encoded matrix has the mapped ground state energy away from $1$. This also explains the Szegedy walk~\cite{szegedy2004quantum}, which quadratically speeds up the classical Markov chain as seen by Ref.~\cite{apers2018quantum}. Overall, the function there is a power function $x^k,~k\in\mathbb{Z}^+$, which satisfies our second criterion, such that a degree $\mathcal{O}(\sqrt{k\log(1/\epsilon)})$ polynomial suffices for an $\epsilon$-approximation~\cite{sachdeva2014faster,gilyen2019quantum}.

We follow the idea of Ref.~\cite{king2025quantum}. We first construct a spectral-amplified operator, which is defined as 
\begin{equation}
    H_{\mathrm{SA}}:=\sum_{j=0}^{r-1}|j\rangle_{\mathrm{a}} \otimes \Pi_j,
\end{equation}
where $a$ denotes the ancillary register. The following lemma gives the block encoding methods of $H_{\mathrm{SA}}$.
\begin{lemma}[{\cite[Lemma 3]{king2025quantum}}]
Given access to the block encoding $U_i$ of $\Pi_j$ and its inverse, there exists a quantum circuit that constructs a block encoding unitary of $H_{\mathrm{SA}}$ with normalization factor $\sqrt{r}$ using $\mathcal{O}(r)$ queries to each $U_i$ and $\mathcal{O}(r)$.
\end{lemma}
\begin{proof}
The construction applies the LCU method and defines: 
\begin{equation}
\begin{aligned}
    \mathrm{ SELECT }:=&\sum_{j=0}^{r-1}|j\rangle\left\langle\left. j\right|_{\mathrm{b}} \otimes\right. U_j,\\
    \mathrm{PREPARE}|0\rangle_{\mathrm{b}} &\mapsto|\alpha\rangle_{\mathrm{b}}:=\frac{1}{\sqrt{r}} \sum_{j=0}^{r-1} |j\rangle_{\mathrm{b}} .
\end{aligned}
\end{equation}
It is readily verified that $\mathrm{ SELECT \cdot PREPARA}$ is the block encoding of $H_{\mathrm{SA}}$:
\begin{equation}
    \left\langle\left. 0\right|_{\mathrm{a}} \mathrm { SELECT } \cdot \mathrm { PREPARE } \mid 0\right\rangle_{\mathrm{b}}|0\rangle_{\mathrm{a}}=\sum_{j=0}^{r-1}|j\rangle_{\mathrm{b}} \otimes \frac{\Pi_j}{\sqrt{r}}.
\end{equation}
\end{proof}

In some sense, $H_{\mathrm{SA}}$ can be seen as the square root of $H_F$, since $H_{\mathrm{SA}}^\dagger H_{\mathrm{SA}}=H_F$. Now, the following lemma gives the final block encoding.
\begin{lemma}[{\cite[Lemma 4]{king2025quantum}}]
Let $U_{\mathrm{SA}}$ be the block encoding of $H_{\mathrm{SA}}$ with normalization factor $\sqrt{r}$. Then, $U_{\mathrm{SA}}^\dagger \left(\mathrm{REF}_{\mathrm{a}} \otimes I\right)U_{\mathrm{SA}}$ is the block encoding unitary of $I-\frac{H_F}{\frac{1}{2}r}$ of normalization factor $1$ and $\mathrm{REF}_{\mathrm{a}}=I_{\mathrm{a}}-2\ket{0}\bra{0}_{\mathrm{a}}$.
\end{lemma}
We find the eigenvalues of the block encoding operator are mapped to $1-2E_i/r$. The spectral gap of the operator is then $2\Delta/r$. Given the above block encoding, the function we need to construct for our targeting $H_F^+$ is
\begin{equation}
    g(x)=\frac{1}{1-x},~x\leq (1-2\Delta/r);~g(x)=0,~x=1.
\end{equation}
Here, we also consider the Moore-Penrose inverse such that when the input approaches $1$, the output becomes zero. When inputting the block encoding operator to the function $g(x)$, we have
\begin{equation}
    \frac{I}{I-\left(I-\frac{P^\perp H_F P^\perp}{\frac{1}{2}r}\right)}=(r/2)^{-1} H_F^+,
\end{equation}
where $P^\perp=I-\ket{\Psi_0}\bra{\Psi_0}$.

The quadratic improved results are provided in the following, and we refer the readers to Ref.~\cite{orsucci2021solving} for a detailed construction.
\begin{lemma}[Adapted from {\cite[Proposition 12]{orsucci2021solving}}]\label{lemma:inverse_ff}
Let $U_F$ be the unitary, which is a $(1,b,0)$-block encoding of $I-2\frac{H_F}{r}$. Then, there exists a quantum circuit that implements the $(K,b+2,\epsilon)$-block encoding of $H_F^+$ with $K\in\Theta(\Delta^{-1})$ using 
\begin{equation}
    \mathcal{O}\left(\sqrt{\frac{r}{\Delta}}\log\left(\frac{r}{\Delta\epsilon}\right) \right)
\end{equation}
queries to $U_F$ and its inverse.
\end{lemma}
\begin{proof}
In \cite[Proposition 12]{orsucci2021solving}, it is required that the block encoded matrix to be $I-\eta A$ with $\eta\in(0,1]$ and $A$ has a spectrum within $[\frac{1}{\kappa},2]$. We take $A$ to be ${H_F}/r$, and when omitting the contribution of the ground state, this determines $\kappa=r/\Delta$. Thus, $\eta=2$. Plugging in the two factors gives the claimed results.
\end{proof}

We note that because $r\geq \norm{H_F}$, the primary improvement from Lemma \ref{lemma:inverse_ff} from \ref{lemma:block_inverse} is that the linear dependence on the upper bound on the operator norm and the inverse of the smallest non-zero eigenvalue is quadratically enhanced. Thus, for estimating the fidelity susceptibility of a frustration-free Hamiltonian, we have the following results.
\begin{theorem}[Formal version of Theorem \ref{theorem:ff}]
Let $\epsilon\in\left(0,\frac{1}{2}\right]$. Let $\chi_F=\norm{G\ket{\Psi_0}}_2^2=\norm{\left(H_F-E_0\right)^{+}H_I\ket{\Psi_0}}_2^2$ with $H_F$ a frustration free Hamiltonian.
There exists an algorithm outputs with probability at least $\frac{8}{\pi^2}$ an $\epsilon$-estimation to $\chi_F$ using 
\begin{equation}
    \mathcal{O}\left(\frac{\sqrt{r}\alpha_I^2}{\Delta^{2.5} \epsilon} \log \left(\frac{r\alpha_I}{\Delta\epsilon}\right)\right)
\end{equation}
queries to $U_F$ and its inverse, and 
\begin{equation}
    \mathcal{O}\left(\frac{\alpha_I^2}{\Delta^2\epsilon} \right)
\end{equation}
queries to $U_I$ and its inverse.
\end{theorem}
\begin{proof}
We replace the construction of the matrix inverse with Lemma \ref{lemma:inverse_ff}, and the rest of the proof is the same as Theorem \ref{theorem:est_fid}.
\end{proof}

\end{document}